\documentclass[11pt]{article}

\usepackage[margin=1in]{geometry}
\usepackage{amsmath}
\usepackage{amssymb}
\usepackage{amsthm}
\usepackage{mathtools}
\usepackage[shortlabels]{enumitem}
\usepackage{url}

\newtheorem{theorem}{Theorem}
\newtheorem{observation}[theorem]{Observation}
\newtheorem{lemma}[theorem]{Lemma}
\newtheorem{corollary}[theorem]{Corollary}
\newtheorem{construction}[theorem]{Construction}

\theoremstyle{definition}
\newtheorem*{definition}{Definition}
\newtheorem{example}[theorem]{Example}

\numberwithin{equation}{section}
\numberwithin{theorem}{section}

\DeclareMathOperator{\Dim}{Dim}
\DeclareMathOperator{\dimH}{dim_H}

\DeclareMathOperator{\dimP}{dim_P}

\DeclareMathOperator{\ldimH}{\underline{dim}_H}
\DeclareMathOperator{\ldimP}{\underline{dim}_P}
\DeclareMathOperator{\udimH}{\overline{dim}_H}
\DeclareMathOperator{\udimP}{\overline{dim}_P}
\DeclareMathOperator{\dimloc}{dim_{loc}}
\DeclareMathOperator{\Dimloc}{Dim_{loc}}
\DeclareMathOperator{\K}{K}

\newcommand{\N}{\mathbb{N}}
\newcommand{\Q}{\mathbb{Q}}
\newcommand{\R}{\mathbb{R}}
\newcommand{\Z}{\mathbb{Z}}
\newcommand{\cA}{\mathcal{A}}
\newcommand{\cB}{\mathcal{B}}
\newcommand{\cP}{\mathcal{P}}
\newcommand{\cQ}{\mathcal{Q}}
\newcommand{\bkappa}{\text{\boldmath $\kappa$}}
\newcommand{\bm}{\mathbf{m}}

\title{Algorithmically Optimal Outer Measures}

\author{Jack H. Lutz\\Iowa State University\and Neil Lutz\\Swarthmore College}

\date{}

\begin{document}

\maketitle

\begin{abstract}
	We investigate the relationship between algorithmic fractal dimensions and the classical local fractal dimensions of outer measures in Euclidean spaces. We introduce global and local optimality conditions for lower semicomputable outer measures. We prove that globally optimal outer measures exist. Our main theorem states that the classical local fractal dimensions of any locally optimal outer measure coincide exactly with the algorithmic fractal dimensions. Our proof uses an especially convenient locally optimal outer measure $\bkappa$ defined in terms of Kolmogorov complexity. We discuss implications for point-to-set principles.
\end{abstract}

\section{Introduction}
	Algorithmic fractal dimensions, which quantify the density of algorithmic information in individual points~\cite{Lutz03b,AHLM07,LutMay08}, have recently been used to prove new theorems~\cite{LutStu20,Lutz21,LuLuMa23,LutStu24} about their classical forerunners, the Hausdorff and packing dimensions of sets. Since algorithmic fractal dimensions are products of the theory of computing, and since the aforementioned new theorems are entirely classical (not involving logic or the theory of computing), these developments call for a more thorough investigation of the relationships between algorithmic and classical fractal dimensions. One significant facet of this investigation, initiated by Orponen~\cite{Orpo21}, is to look for purely classical proofs of these new classical theorems.
	
	In this paper, taking a different approach, we establish direct connections between algorithmic and classical fractal dimensions. Aside from the presence versus absence of algorithms, the most striking difference between algorithmic fractal dimensions and classical fractal dimensions is that the algorithmic dimensions are usefully defined for individual points in Euclidean space, while the classical Hausdorff and packing dimensions vanish on individual points. To bridge this gap, we examine the classical {\it local dimensions} (also called \emph{pointwise dimensions}) of outer measures at individual points in Euclidean spaces~\cite{Falc14}. These local fractal dimensions have been studied at least since the 1930s and are essential tools in multifractal analysis~\cite{Fros35, Falc97}. Outer measures and the algorithmic and local dimensions are defined precisely in Section~\ref{sec:alfd} below.
	
	Outer measures, introduced by Carath\'{e}odory~\cite{Cara14} in the ``prehistory'' of Hausdorff dimension~\cite{Haus19} (defining what later became known as the 1-dimensional Hausdorff measure), are now best known for their role in Carath\'{e}odory's program~\cite{Cara18} to generalize Lebesgue measure to a wide variety of settings~\cite{Tao11}. However, it is the role of outer measures in local fractal dimensions that are of interest here.
	
	The second author observed~\cite{Lutz17} that a particular, very nonclassical outer measure $\bkappa$, defined in terms of Kolmogorov complexity, has the property that the classical local fractal dimensions of $\bkappa$ coincide exactly with the algorithmic fractal dimensions at every point in $\R^n$. This property of $\bkappa$ is analogous to Levin's coding theorem~\cite{jLevi73,jLevi74}, which pertains to a particular, very nonclassical subprobability measure $\bm$ on strings. Levin's theorem says that if we substitute $\bm$ for the probability measure $p$ in the classical Shannon \emph{self-information}~\cite{Shan48} $\log 1/p(x)$, then the resulting quantity $\log 1/\bm(x)$ is essentially the prefix Kolmogorov complexity (i.e., the algorithmic information content) of the string $x$.
	
	Levin defined $\bm$ as an optimal lower semicomputable subprobability measure, so the above analogy leads us to investigate here the algorithmic optimality properties of $\bkappa$ and other outer measures on Euclidean spaces.
	
	We first investigate outer measures that are \emph{globally optimal}, a property that is closely analogous to the optimality property of Levin's $\bm$. In Section~\ref{sec:gao} we prove \emph{ab initio} that globally optimal outer measures on $\R^n$ exist. We also show that Levin's $\bm$ lifts in a natural way to a specific example of such a measure.
	
	As it turns out, the outer measure $\bkappa$ is not globally optimal. In Section~\ref{sec:lao} we prove this fact, and we introduce and investigate the more general and more subtly defined class of \emph{locally optimal} outer measures on $\R^n$. Our main theorem establishes that \emph{every} locally optimal outer measure $\mu$ on a Euclidean space $\R^n$ has the property that the classical local fractal dimensions of $\mu$ coincide exactly with the algorithmic dimensions at every point in $\R^n$.
	
	In Section~\ref{sec:pspdm} we discuss implications of our results, especially for the point-to-set principles that have enabled the new classical results mentioned in the first paragraph of this introduction.

\section{Algorithmic and Local Fractal Dimensions}\label{sec:alfd}
	This section reviews the algorithmic fractal dimensions and the classical local fractal dimensions.

	Following standard practice~\cite{Nies09,DowHir10,LiVit19}, we fix a universal prefix Turing machine $U$ and define the \emph{(prefix) Kolmogorov complexity} of a string $w\in\{0,1\}^*$ to be
	\[\K(w)=\min\left\{|\pi|\;\big|\;\pi\in\{0,1\}^*\text{ and }U(\pi)=w\right\},\]
	i.e., the minimum number of bits required to cause $U$ to output $w$. By standard binary encodings, we extend this from $\{0,1\}^*$ to other countable domains. In particular, the Kolmogorov complexity $\K(q)$ of a rational point $q\in\Q^n$ is well defined.

	The \emph{Kolmogorov complexity} of a point $x\in \R^n$ at a \emph{precision} $r\in\N$ is
	\[\K_r(x)=\min\left\{\K(q)\mid q\in \Q^n\text{ and }|q-x|< 2^{-r}\right\},\]
	where $|q-x|$ is the Euclidean distance from $q$ to $x$.

	We now define the algorithmic fractal dimensions of points in $\R^n$.
	\begin{definition}[\cite{Lutz03b,Mayo02,AHLM07}] Let $x\in \R^n$.
		\begin{enumerate}
			\item The \emph{algorithmic dimension} of $x$ is
			\begin{equation}\label{eq:dim}
				\dim(x)=\liminf_{r\to\infty}\frac{\K_r(x)}{r}.
			\end{equation}
			\item The \emph{strong algorithmic dimension} of $x$ is
			\begin{equation}\label{eq:Dim}
				\Dim(x)=\limsup_{r\to\infty}\frac{\K_r(x)}{r}.
			\end{equation}
		\end{enumerate}
	\end{definition}
	See~\cite{Reim04,Mayo08,LutLut20} for surveys of these notions.

	The classical local fractal dimensions are local properties of outer measures. An \emph{outer measure} on a set $X$~\cite{Tao11} is a function $\mu:\cP(X)\to[0,\infty]$ (where $\cP(X)$ is the power set of $X$) with the following three properties.
	\begin{enumerate}[(i)]
		\item (vanishes on empty set) $\mu(\emptyset)=0$.
		\item (monotonicity) For all $E,F\subseteq X$,
		\[E\subseteq F\implies \mu(E)\leq \mu(F).\]
		\item (countable subadditivity) For all $E_0,E_1,\ldots\subseteq X$,
		\[\mu\left(\bigcup_{n=0}^\infty E_n\right)\leq \sum_{n=0}^\infty \mu(E_n).\]
	\end{enumerate}
	An outer measure $\mu$ on a set $X$ is \emph{finite} if $\mu(X)<\infty$ and \emph{locally finite} if, for every point $x\in X$, there is an open set $E\ni x$ such that $\mu(E)<\infty$.

	\begin{definition}[\cite{Falc14}]
		If $\mu$ is a locally finite outer measure on $\R^n$, then the \emph{lower} and \emph{upper local} (or \emph{pointwise}) \emph{dimensions} of $\mu$ at a point $x\in \R^n$ are
		\begin{equation}\label{eq:dimloc}
			\dimloc\mu(x)=\liminf_{r\to\infty}\frac{\log\frac{1}{\mu(B(x,2^{-r}))}}{r}
		\end{equation}
		and
		\begin{equation}\label{eq:Dimloc}
			\Dimloc\mu(x)=\limsup_{r\to\infty}\frac{\log\frac{1}{\mu(B(x,2^{-r}))}}{r},
		\end{equation}
		respectively. (The logarithms here are base-2, and $B(x,\varepsilon)$ is the open ball of radius $\varepsilon$ about $x$ in $\R^n$.)
	\end{definition}
	
	As stated in the introduction, our main objective is to identify a class of outer measures that cause the classical local fractal dimensions~\eqref{eq:dimloc} and~\eqref{eq:Dimloc} to coincide with the algorithmic fractal dimensions~\eqref{eq:dim} and~\eqref{eq:Dim}.

\section{Global Algorithmic Optimality}\label{sec:gao}

	The optimality notions that we discuss in this paper concern outer measures with three special properties that we now define.

	\begin{definition}
		An outer measure $\mu$ on $\R^n$ is \emph{finitely supported} on $\Q^n$ if, for every $\varepsilon>0$, there is a finite set $A\subseteq\Q^n$ such that $\mu(\R^n\setminus A)<\varepsilon$.
	\end{definition}

	Note that an outer measure $\mu$ on $\R^n$ that is finitely supported on $\Q^n$ is supported on $\Q^n$ in the usual sense that $\mu(\R^n\setminus\Q^n)=0$. The following example shows that the converse does not hold.

	\begin{example}\label{ex:suppnotfinsupp}
		Define the function $\mu:\cP(\R^n)\to[0,1]$ by
		\[\mu(E)=1-2^{-|E\cap\Q^n|},\]
		where $2^{-\infty}=0$. This function is an outer measure on $\R^n$. To see that it is countably subadditive, let $E_0,E_1,\ldots\in\R^n$, let $E=\bigcup_{n=0}^\infty E_i$, and consider two cases.
		\begin{enumerate}
			\item If there is some $j\in\N$ with $|E_j\cap\Q^n|=|E\cap\Q^n|$, then \[\mu(E)=\mu(E_j)\leq\sum_{i=0}^\infty \mu(E_i).\]
			\item Otherwise, there are some distinct $j,k\in\N$ such that $E_j\cap\Q^n$ and $E_k\cap\Q^n$ are both non-empty. In this case, we have
			\[\mu(E)\leq 1\leq \mu(E_j)+\mu(E_k)\leq\sum_{i=0}^\infty \mu(E_i).\]
		\end{enumerate}
		Furthermore, $\mu$ is supported, but not finitely supported, on $\Q^n$.
	\end{example}

	\begin{definition}
		An outer measure $\mu$ on $\R^n$ is \emph{strongly finite} if $\mu$ is supported on $\Q^n$ and
		\[\sum_{q\in\Q^n}\mu(\{q\})<\infty.\]
	\end{definition}
	It is clear that every strongly finite outer measure is finite. The outer measure of Example~\ref{ex:suppnotfinsupp} shows that the converse does not hold. Note also that every strongly finite outer measure is finitely supported on $\Q^n$.

	\begin{definition}
		An outer measure $\mu$ on $\R^n$ is \emph{lower semicomputable} if it is finitely supported on $\Q^n$ and there is a computable function
		\[\hat{\mu}:\cP^{<\omega}(\Q^n) \times\N\to\Q\cap[0,\infty)\]
		(where $\cP^{<\omega}(\Q^n)$ is the \emph{finite power set} of $\Q^n$, i.e, the set of all finite subsets of $\Q^n$) with the following two properties.
		\begin{enumerate}[(i)]
			\item For all $A\in\cP^{<\omega}(\Q^n)$ and $s,t\in\N$,
			\[s\leq t\implies \hat{\mu}(A,s)\leq \hat{\mu}(A,t)\leq \mu(A).\]
			\item For all $A\in \cP^{<\omega}(\Q^n)$,
			\[\lim_{t\to\infty}\hat{\mu}(A,t)=\mu(A).\]
		\end{enumerate}
	\end{definition}

	\begin{observation}\label{obs:thirdcond}
		Given any lower semicomputable outer measure $\mu$, there is a computable function $\hat{\mu}:\cP^{<\omega}(\Q^n) \times\N\to\Q\cap[0,\infty)$ that satisfies properties (i) and (ii) from the definition of lower semicomputability and also satisfies, for all $t\in\N$ and $A,B\in\cP^{<\omega}(\Q^n)$,
		\[A\subseteq B\implies \hat{\mu}(A,t)\leq\hat{\mu}(B,t).\]
	\end{observation}
	\begin{proof}
		Let $\hat{\nu}$ be a function testifying to the lower semicomputability of $\mu$, and define
		\[\hat{\mu}(A,t)=\max_{B\subseteq A}\hat{\nu}(B,t).\]
	\end{proof}

	\begin{definition}
		An outer measure $\mu$ on $\R^n$ is \emph{globally optimal} if the following properties hold.
		\begin{enumerate}[(i)]
			\item $\mu$ is strongly finite and lower semicomputable.
			\item For every strongly finite, lower semicomputable outer measure $\nu$ on $\R^n$, there is a constant $\beta\in(0,\infty)$ such that, for all $E\subseteq\R^n$,
			\[\mu(E)\geq\beta\cdot\nu(E).\]
		\end{enumerate}
	\end{definition}

    Following standard practice in the analogous discrete setting~\cite{LiVit19,DowHir10,SheVer02}, we first give a direct construction of a globally optimal outer measure on $\R^n$. In constructing this outer measure, we will use a computable enumeration of all strongly finite, lower semicomputable outer measures that take values in $[0,1]$. We now describe this enumeration.

 	\begin{construction}\label{con:theta}
		Fix a computable enumeration $q_0,q_1,q_2,\ldots$ of $\Q^n$, and for all $t\in \N$, let $\Q^n_t$ denote $\{q_0,\ldots,q_t\}$.
		Let $M_0,M_1,M_2\ldots$ be a computable enumeration of all prefix Turing machines that take inputs in $\cP^{<\omega}(\Q^n) \times\N$, give outputs (if they halt) in $\Q\cap[0,1]$, and satisfy
		\begin{equation}\label{eq:nullset}
			M_k(\emptyset,t)=0
		\end{equation}
		for all $k,t\in\N$. For all $k\in\N$, we define the following functions.
		\begin{itemize}
			\item The function $\tau_k:\N\times\N\to\{0,1\}$ is given by
			\begin{equation*}\label{eq:taudef}
				\tau_k(s,t)=\begin{cases}1&\text{if for all }D\subseteq \Q^n_s\text{ and }r\leq s\text{, }M_k\text{ halts on input }(D,r)\text{ within }t\text{ steps}\\0&\text{otherwise.}\end{cases}.
			\end{equation*}
			Observe that $\tau_k(s,t)$ is a computable function of $k$, $s$, and $t$, and that $\tau_k$ is monotone in its second argument. This function will serve as an indicator for sufficiently large wait parameters. The parameter $b$ is present to ensure strong finiteness.
			\item For all $b\in\N$, the function $\eta_{k,b}:\cP^{<\omega}(\Q^n) \times\N\to\Q\cap[0,1]$ is given by
			\begin{equation}\label{eq:etadef}
				\eta_{k,b}(A,t)=\max_{\substack{C\subseteq A\\s\leq t}}\left\{M_k(C\cap\Q^n_s,s)\;\middle|\; \tau_k(s,t)=1\text{ and }\sum_{q\in\Q^n_s}M_k(\{q\},s)\leq b\right\}.
			\end{equation}
			It follows from the computability and monotonicity of $\tau_k$ that $\eta_{k,b}$ is also computable and is monotone in both arguments, and it follows from~\eqref{eq:nullset} that $\eta_{k,b}(\emptyset,t)=0$ for all $t\in\N$. Optimizing over finite covers of $A$ yields countable subadditivity.
			\item For all $b\in\N$, the function $\hat{\theta}_{k,b}:\cP^{<\omega}(\Q^n) \times\N\to\Q\cap[0,1]$ is given by
			\begin{equation}\label{eq:thetahatdef}
				\hat{\theta}_{k,b}(A,t)=\min\left\{\sum_{i=0}^{\ell} \eta_{k,b}(A_i,t)\;\middle|\;\ell\in\N,\ A_0,\ldots,A_{\ell}\subseteq A,\ \text{and}\ A\subseteq \bigcup_{i=0}^{\ell} A_i\right\}.
			\end{equation}
			From $\eta_{k,b}$, this function inherits computability, monotonicity in both arguments, and the property $\hat{\theta}_{k,b}(\emptyset,t)=0$ for all $t\in\N$.
	
			\item For all $b\in\N$, the function $\theta_{k,b}:\cP (\R^n)\to[0,1]$ is given by
			\begin{equation}\label{eq:thetadef}
				\theta_{k,b}(E)=\lim_{t\to\infty}\hat{\theta}_{k,b}(E\cap\Q^n_t,t).
			\end{equation}
			This limit exists by the monotonicity of $\hat{\theta}_{k,b}$ in its second argument.
		\end{itemize}
		\end{construction}
		\begin{lemma}\label{lem:strongfin}
			For all $k,b\in\N$, the function $\theta_{k,b}$ is a strongly finite, lower semicomputable outer measure.
		\end{lemma}
		\begin{proof}
			Fix $k,b\in\N$. The function $\theta_{k,b}$ is monotone by the monotonicity of $\hat{\theta}_{k,b}$ in its first argument, and it vanishes on the empty set because $\hat{\theta}_{k,b}(\emptyset,t)=0$ for all $t\in\N$. To prove that $\theta_{k,b}$ is an outer measure, then, it suffices to show that $\theta_{k,b}$ is countably subadditive.
		
            For this, let $E_0,E_1,E_2,\ldots\subseteq \R^n$, and let $E=\bigcup_{i=0}^\infty E_i$. Let $\varepsilon>0$, and let $A\in\cP^{<\omega}(E\cap\Q^n)$ and $t\in\N$ be such that
            \begin{equation}\label{eq:mainlem1}
                \theta_{k,b}(E)<\hat{\theta}_{k,b}(A,t)+\varepsilon.
            \end{equation}
            For each $i\in\N$, let $A_i=E_i\cap A$, and let $\ell_i$ and $A_{i,0},\ldots,A_{i,\ell_i-1}$ be such that
            \begin{equation}\label{eq:mainlem2}
                \hat{\theta}_{k,b}(A_i,t)=\sum_{j=0}^{\ell_i-1}\eta_{k,b}(A_{i,j},t).
            \end{equation}
            Since $A$ is a finite set, there is some $\ell\in\N$ such that $A\subseteq \bigcup_{i=0}^{\ell-1}A_i$, so
            \begin{align*}
                \theta_{k,b}(E)-\varepsilon&<\hat{\theta}_{k,b}(A,t)\tag{by inequality~\eqref{eq:mainlem1}}\\
                &\leq\sum_{i=0}^{\ell-1}\sum_{j=0}^{\ell_i-1}\eta_{k,b}(A_{i,j},t)\tag{by the definition of $\hat{\theta}_{k,b}(A,t)$}\\
                &=\sum_{i=0}^{\ell-1}\hat{\theta}_{k,b}(A_i,t)\tag{by equation~\eqref{eq:mainlem2}}\\
                &\leq \sum_{i=0}^{\ell-1}\theta_{k,b}(E_i).\tag{by the monotonicity of $\hat{\theta}_{k,b}$}
            \end{align*}
            Letting $\varepsilon\to 0$, we have
            \[\theta_{k,b}(E)\leq \sum_{i=0}^\infty\theta_{k,b}(E_i),\]
            and we conclude that $\theta_{k,b}$ is an outer measure.

			To see that $\theta_{k,b}$ is strongly finite, observe that
			\[\sum_{q\in\Q^n}\theta_{k,b}(\{q\})=\sum_{q\in\Q^n}\lim_{t\to\infty}\hat{\theta}_{k,b}(\{q\}\cap\Q^n_t,t)=\lim_{t\to\infty}\sum_{q\in\Q^n}\hat{\theta}_{k,b}(\{q\}\cap\Q^n_t,t)\]
			by~\eqref{eq:thetadef} and the monotone convergence theorem. Hence, it suffices to show for all $t\in\N$ that
			\begin{equation*}\label{eq:bbound1}
				\sum_{q\in\Q^n}\hat{\theta}_{k,b}(\{q\}\cap\Q^n_t,t)\leq b.
			\end{equation*}
			To this end, fix $t\in\N$. As noted above,
			\begin{equation}\label{eq:nulleta}
				\hat{\theta}_{k,b}(\emptyset,t)=0.
			\end{equation}
			For all $q\in\Q^n$, we have
			\begin{equation}\label{eq:singleton}
				\hat{\theta}_{k,b}(\{q\},t)=\eta_{k,b}(\{q\},t).
			\end{equation}
			By~\eqref{eq:etadef}, there is some $s\leq t$ such that
			\begin{equation}\label{eq:etabound}
				\sum_{q\in\Q^n}\eta_{k,b}(\{q\}\cap\Q^n_t,t)=\sum_{q\in\Q^n}M_k(\{q\}\cap\Q^n_t\cap\Q^n_s,s)
				=\sum_{q\in\Q^n_s}M_k(\{q\},s)
				\leq b.
			\end{equation}
			Combining~\eqref{eq:nulleta},~\eqref{eq:singleton} and~\eqref{eq:etabound}, we have
			\begin{align*}
				\sum_{q\in\Q^n}\hat{\theta}_{k,b}(\{q\}\cap\Q^n_t,t)&=\sum_{q\in\Q^n_t}\hat{\theta}_{k,b}(\{q\},t)\\
				&=\sum_{q\in\Q^n_t}\eta_{k,b}(\{q\},t)\\
				&\leq b.
			\end{align*}

			Finally, $\hat{\theta}_{k,b}$ is a witness to the semicomputability of $\theta_{k,b}$; equation~\eqref{eq:thetadef} and the monotonicity of $\hat{\theta}_{k,b}$ in both arguments give us property (i) of lower semicomputable outer measures, and property (ii) of lower semicomputable outer measures is immediate from equation~\eqref{eq:thetadef}.
		\end{proof}

	\begin{lemma}\label{lem:maintech}
		Let $\Theta$ be the set of all strongly finite, lower semicomputable outer measures $\mu:\cP(\R^n)\to[0,1]$. Then $\Theta=\{\theta_{k,b}\mid k,b\in\N\}$.
	\end{lemma}
	\begin{proof}
		Lemma~\ref{lem:strongfin} establishes that $\{\theta_{k,b}\mid k,b\in\N\}\subseteq\Theta$. For the other direction, let $\mu\in\Theta$, and let $\hat{\mu}$ be a witness to the lower semicomputability of $\mu$. By Observation~\ref{obs:thirdcond} we can assume without loss of generality that $\hat{\mu}$ is monotone in its first argument. Then $\hat{\mu}$ is computable, and $\hat{\mu}(\emptyset,t)=0$ for all $t\in\N$. Therefore there is also some $k\in\N$ such that, for all $(A,t)\in\cP^{<\omega}(\Q^n) \times\N$, we have $M_k(A,t)=\hat{\mu}(A,t)$.

		We now show that for all $A\in\cP^{<\omega}(\Q^n) $, \[\lim_{t\to\infty}\hat{\theta}_{k,b}(A,t)=\lim_{t\to\infty}\hat{\mu}(A,t).\]
		Let $A\in\cP^{<\omega}(\Q^n) $ and $t\in\N$, and let $b=\sum_{q\in\Q^n}\mu(\{q\})<\infty$. Then
		\begin{align*}
			\hat{\theta}_{k,b}(A,t)&\leq \eta_{k,b}(A,t)\\
			&=\max_{\substack{C\subseteq A\\s\leq t}}\left\{M_k(C\cap\Q^n_s,s)\;\middle|\; \tau_{k}(s,t)=1\text{ and }\sum_{q\in\Q^n_s}M_k(\{q\},s)\leq b\right\}\\
			&\leq \max_{\substack{C\subseteq A\\s\leq t}}M_{k}(C\cap\Q^n_s,s)\\
			&= \max_{\substack{C\subseteq A\\s\leq t}}\hat{\mu}(C\cap\Q^n_s,s)\\
			&\leq\hat{\mu}(A,t),
		\end{align*}
		by the monotonicity of $\hat{\mu}$.

		Now let $\varepsilon>0$, let $t\in\N$ be such that $A\subseteq\Q^n_t$ and for all $C\subseteq A$,
		\[\hat{\mu}(C,t)\geq\mu(C)-\varepsilon/2^{|A|},\]
		and let $T\geq t$ be sufficiently large so that $\tau_k(t,T)=1$. Such a $T$ exists because the computability of $\hat{\mu}$ implies that $M_k$ halts on all inputs in $\mathcal{P}^{<\omega}(\Q^n)\times\N$. Then for each $C\subseteq A$,
		\begin{align*}
			\eta_{k,b}(C,T)
			&=\max_{\substack{D\subseteq C\\s\leq T}}\left\{M_k(D\cap\Q^n_s,s)\;\middle|\; \tau_k(s,T)=1\text{ and }\sum_{q\in\Q^n_s}M_k(\{q\},s)\leq b\right\}\\
			&\geq M_k(C\cap \Q^n_t,t)\\
			&=\hat{\mu}(C,t).
		\end{align*}
		It follows that
		\begin{align*}
			\hat{\theta}_{k,b}(A,T)&=\min\left\{\sum_{i=0}^{\ell} \eta_k(A_i,T)\;\middle|\;\ell\in\N,\ A_0,\ldots,A_{\ell}\subseteq A\text{, and }A\subseteq \bigcup_{i=0}^{\ell} A_i\right\}\\
			&\geq \min\left\{\sum_{i=0}^{\ell} \hat{\mu}(A_i,t)\;\middle|\;\ell\in\N,\ A_0,\ldots,A_{\ell}\subseteq A\text{, and }A\subseteq \bigcup_{i=0}^{\ell} A_i\right\}\\
			&\geq \min\left\{\sum_{i=0}^{\ell} \left(\mu(A_i)-\frac{\varepsilon}{2^{|A|}}\right)\;\middle|\;\ell\in\N,\ A_0,\ldots,A_{\ell}\subseteq A\text{, and }A\subseteq \bigcup_{i=0}^{\ell} A_i\right\}\\
			&\geq\mu(A)-\varepsilon\\
			&\geq\hat{\mu}(A,t)-\varepsilon,
		\end{align*}
		by the countable subadditivity of $\mu$.

		We have shown that for every $A\in\cP^{<\omega}(\Q^n) $, every $\varepsilon>0$, and every sufficiently large $t\in\N$, $\hat{\theta}_{k,b}(A,t)\leq\hat{\mu}(A,t)$ and there exists some $T\in\N$ such that $\hat{\theta}_{k,b}(A,T)\geq \hat{\mu}(A,t)-\varepsilon$. This implies that for all $A\in\cP^{<\omega}(\Q^n) $,
		\[\lim_{t\to\infty}\hat{\theta}_{k,b}(A,t)=\lim_{t\to\infty}\hat{\mu}(A,t),\]
		and therefore for all $E\subseteq \R^n$,
		\begin{align*}
			\theta_{k,b}(E)
			&=\lim_{t\to\infty}\hat{\theta}_{k,b}(E\cap\Q^n_t,t)\\
			&=\lim_{t\to\infty}\hat{\mu}(E\cap\Q^n_t,t)\\
			&=\mu(E).
		\end{align*}
		We conclude that $\mu\in\{\theta_{k,b}\mid k,b\in\N\}$, so $\Theta\subseteq \{\theta_{k,b}\mid k,b\in\N\}$.
	\end{proof}

	\begin{theorem}
		Globally optimal outer measures exist.
	\end{theorem}

	\begin{proof}
		Let $f:\N\times\N\to\N$ be any computable injective monotone function, and define the strongly finite outer measure $\theta:\cP(\R^n)\to[0,1]$ by
		\[\theta(E)=\sum_{k,b\in\N}\frac{\theta_{k,b}(E)}{2^{f(k,b)+1}},\]
		where $\theta_{k,b}$ is defined as in Construction~\ref{con:theta}. This outer measure is supported on $\cP^{<\omega}(\Q^n)$, and the function $\hat{\theta}:\cP^{<\omega}(\Q^n)\times\N\to[0,1]$ given by
		\[\hat{\theta}(A,t)=\sum_{k,b\leq t}\frac{\hat{\theta}_{k,b}(A,t)}{2^{f(k,b)+1}}\]
		is a witness to the lower semicomputability of $\theta$. Hence, $\theta\in\Theta$.

		Let $\mu:\cP(\R^n)\to[0,\infty)$ be any strongly finite, lower semicomputable outer measure, and let \[h=\lceil\mu(\R^n)\rceil.\]
		Then the function $\tilde{\mu}:\cP(\R^n)\to[0,1]$ given by
		\[\tilde{\mu}(E)=\mu(E)/h\]
		is a strongly finite, lower semicomputable outer measure. By Lemma~\ref{lem:maintech}, there are some $k,b\in\N$ such that, for all $E\subseteq \R^n$,
		\[\theta_{k,b}(E)=\tilde{\mu}(E).\]
		We have, for all $E\subseteq\R^n$,
		\[\mu(E)=h\cdot \tilde{\mu}(E)= h\cdot\theta_{k,b}(E)\leq h\cdot 2^{f(k,b)+1}\cdot \theta(E),\]
		so $\theta$ is globally optimal.
	\end{proof}

    We thank an anonymous referee for pointing out the following specific example of a globally optimal outer measure on $\R^n$.
    
    Just as we have ``lifted'' Kolmogorov complexity from $\{0,1\}^*$ to $\Q^n$ via routine encoding, we lift Levin's optimal lower semicomputable subprobability measure $\bm$~\cite{jLevi73,jLevi74} from $\{0,1\}^*$ to $\Q^n$ so that, for all $q\in\Q^n$,
    \[\bm(q)=\sum\limits_{\substack{\pi\\U(\pi)=q}}2^{-|\pi|},\]
    and we define $\mathbf{m}:\mathcal{P}(\R^n)\to[0,1]$ by
    \[\bm(E)=\sum_{q\in E\cap\Q^n}\bm(q).\]
    \begin{theorem}
        $\mathbf{m}:\mathcal{P}(\R^n)\to[0,1]$ is a globally optimal outer measure. 
    \end{theorem}
    \begin{proof}
        It is clear that $\mathbf{m}$ is a strongly finite, lower semicomputable outer measure on $\R^n$. For global optimality, let $\nu$ be any strongly finite, lower semicomputable outer measure on $\R^n$. Then there is some $b$ such that $\sum_{q\in\Q^n}\nu(q)\leq b$, and $\nu/b$ is a subprobability measure. Therefore there is some constant $\alpha$ such that for all $E\subseteq\R^n$,
        \begin{align*}
            \nu(E)&=\nu(E\cap\Q^n)\tag{$\nu$ is supported on $\Q^n$}\\
            &\leq\sum_{q\in E\cap\Q^n}\nu(q)\tag{$\nu$ is countably subadditive}\\
            &\leq \sum_{q\in E\cap\Q^n}\alpha b\mathbf{m}(q)\tag{optimality of $\mathbf{m}$}\\
            &=\alpha b\mathbf{m}(E).\qedhere
        \end{align*}
    \end{proof}
\section{Local Algorithmic Optimality}\label{sec:lao}

	This paper's investigation of algorithmic optimality is primarily driven by a specific outer measure $\bkappa$. To define $\bkappa$, we first define the \emph{Kolmogorov complexity of a set} $E\subseteq\R^n$ to be
	\[\K(E)=\min\{\K(q)\mid q\in E\cap\Q^n\},\]
	where $\min\emptyset=\infty$. That is, $\K(E)$ is the minimum number of bits required to cause the universal prefix Turing machine $U$ to print \emph{some} rational point in $E$. (Shen and Vereschagin~\cite{SheVer02} introduced a similar notion for a different purpose.)

	\begin{definition}[\cite{Lutz17}]
		Define the function $\bkappa:\cP(\R^n)\to[0,1]$ by
		\[\bkappa(E)=2^{-\K(E)},\]
		where $2^{-\infty}=0$, for all $E\subseteq\R^n$.
	\end{definition}

	\begin{observation}[\cite{Lutz17}]
		$\bkappa$ is an outer measure on $\R^n$.
	\end{observation}

	Our primary interest in $\bkappa$ is the following connection between classical local fractal dimensions and algorithmic fractal dimensions.

	\begin{observation}[\cite{Lutz17}]\label{obs:kappaloc}
		For all $x\in \R^n$,
		\[\dimloc\bkappa(x)=\dim(x)\]
		and
		\[\Dimloc\bkappa(x)=\Dim(x).\]
	\end{observation}
	\begin{proof}
		By \eqref{eq:dim}--\eqref{eq:Dimloc}, it suffices to note that, for all $x\in\R^n$,
		\[\log\frac{1}{\bkappa(B(x,2^{-r}))}=K_r(x).\qedhere\]
	\end{proof}
	
	We next investigate the algorithmic properties of the outer measure $\bkappa$.
	\begin{observation}\label{obs:kappasflc}
		$\bkappa$ is strongly finite and lower semicomputable.
	\end{observation}
	\begin{proof}
		It suffices to show three things.
		\begin{enumerate}
			\item
				$\bkappa$ is finitely supported on $\Q^n$. For this, let $\varepsilon>0$. Let
				\[A=\left\{q\in\Q^n\;\middle|\; \K(q)\leq \log\frac{1}{\varepsilon}\right\}.\]
				Then $A$ is a finite subset of $\Q^n$, and
				\begin{align*}
					\K(\R^n\setminus A)&=\min\{\K(q)\mid q\in \Q^n\setminus A\}\\
					&>\log\frac{1}{\varepsilon},
				\end{align*}
				so $\bkappa(\R^n\setminus A)<\varepsilon$.
			\item
				$\displaystyle \sum_{q\in\Q^n}\bkappa(\{q\})<\infty$. For this, just note that
				\[\sum_{q\in\Q^n}\bkappa(\{q\})=\sum_{q\in\Q^n}2^{-\K(q)}\leq 1,\]
				by the Kraft inequality for prefix Kolmogorov complexity.
			\item
				$\bkappa$ is lower semicomputable. This follows immediately from the well known upper semicomputability of the Kolmogorov complexity function.
		\end{enumerate}
	\end{proof}

	\begin{lemma}\label{lem:kappanotgo}
		$\bkappa$ is not globally optimal.
	\end{lemma}
	\begin{proof}
		Define the function $\nu:\cP(\R^n)\to[0,\infty]$ by
		\[\nu(E)=\sum_{q\in E\cap\Q^n}2^{-\K(q)}.\]
		We now show that $\nu$ is a strongly finite, lower semicomputable outer measure on $\R^n$. It is clear that $\nu$ is an outer measure on $\R^n$. It thus suffices to prove that $\nu$ has the properties 1, 2, and 3 proven for $\bkappa$ in the proof of Observation~\ref{obs:kappasflc}. For properties 2 and 3, the proofs for $\nu$ are identical to those for $\bkappa$. For property 1, that $\nu$ is finitely supported on $\Q^n$, let $\varepsilon>0$. By the Kraft inequality for prefix Kolmogorov complexity,
		\[\sum_{q\in\Q^n}2^{-\K(q)}\leq 1,\]
		so there is a finite set $A\subseteq\Q^n$ such that
		\[\nu(\R^n\setminus A)=\sum_{q\in\Q^n\setminus A}2^{-\K(q)}<\varepsilon.\]
		Hence, to prove the lemma, it suffices to exhibit, for all $\beta\in(0,\infty)$, a set $E\subseteq\R^n$ such that, 
		\begin{equation}\label{eq:kappanotgo}
			\bkappa(E)<\beta\nu (E).
		\end{equation}

		Let $\beta\in(0,\infty)$. Let $c$ be a constant such that, for all $m\in\N$,
		\[\K(m,\underbrace{0,\ldots,0}_{n-1})< \log(m)+2\log\log(m)+c\,.\]
		Let $\alpha\in(0,\infty)$ be some parameter, let $\gamma=2+2\log(\alpha+2)+c$,
		and define the set
		\[E=\left\{q\in\Q^n\;\middle|\;\K(q)>\alpha\right\}.\]
		Then $\bkappa(E)< 2^{-\alpha}$, and
		\begin{align*}
			\nu (E)&= \sum_{q\in E\cap\Q^n}2^{-\K(q)}\\
			&\geq\sum\limits_{\substack{q\in\Q^n\\\K(q)\in(\alpha,\alpha+\gamma)}}2^{-\K(q)}\\
			&\geq 2^{-\alpha-\gamma}\cdot|\{q\in\Q^n\mid \K(q)\in(\alpha,\alpha+\gamma)\}|\,.
		\end{align*}
		There are fewer than $2^{\alpha+1}$ rational points $q$ with $\K(q)\leq\alpha$. For all $m\in\N$ such that $m\leq 2^{\alpha+2}$,
		\begin{align*}
			\K(m,0,\ldots,0)
			&< \alpha+2+2\log(\alpha+2)+c\\
			&=\alpha+\gamma\,,
		\end{align*}
		so there are at least $2^{\alpha+1}$ rational points $q$ with $\K(q)\in(\alpha,\alpha+\gamma)$, and we have $\nu (E)> 2^{-\gamma}$.
		Thus,
		\[\bkappa(E)<2^{\gamma-\alpha}\nu (E).\]
		Choosing $\alpha$ such that
		\[\gamma-\alpha=2+2\log(\alpha+2)+c-\alpha<\log\beta\]
		yields~\eqref{eq:kappanotgo}.
	\end{proof}

	Notwithstanding Lemma~\ref{lem:kappanotgo}, $\bkappa$ does have an optimality property, which we next define. For each $m=(m_1,\ldots,m_n)\in\Z^n$, let
	\[Q_m=[m_1,m_1+1)\times\cdots\times[m_n,m_n+1)\]
	be the \emph{unit cube} at $m$. For each such $m$ and each $r\in\N$, let
	\[Q^{(r)}_m=2^{-r}Q_m=\left\{2^{-r}x\;\middle|\; x\in Q_m\right\}\]
	be the \emph{$r$-dyadic cube} with \emph{address} $m$. Note that each $Q^{(r)}_m$ is ``half-closed, half-open'' in such a way that, for each $r\in\N$, the family
	\[\cQ^{(r)}=\left\{Q^{(r)}_m\;\middle|\;m\in\Z^n\right\}\]
	is a partition of $\R^n$.

	\begin{definition}
		Let $\mu$ and $\nu$ be outer measures on $\R^n$, and let $\cA=(\cA^{(r)}\mid r\in\N)$ be a sequence of families $\cA^{(r)}\subseteq\cP(\R^n)$ of subsets of $\R^n$. We say that $\mu$ \emph{dominates} $\nu$ on $\cA$ if there is a function $\beta:\N\to(0,\infty)$ such that $\beta(r)=2^{-o(r)}$ as $r\to\infty$ and, for every $r\in\N$ and every set $E\in\cA^{(r)}$,
		\[\mu(E)\geq\beta(r)\cdot\nu(E).\]

		We say that $\mu$ \emph{dominates $\nu$ on dyadic cubes} if $\mu$ dominates $\nu$ on $\cQ=(\cQ^{(r)}\mid r\in\N)$. We say that $\mu$ \emph{dominates $\nu$ on balls} if $\mu$ dominates $\nu$ on $\cB=(\cB^{(r)}\mid r\in\N)$, where $\cB^{(r)}$ is the set of all open balls of radius $2^{-r}$ in $\R^n$.
	\end{definition}
	
	\begin{definition}
		An outer measure $\mu$ on $\R^n$ is \emph{locally optimal} if the following two properties hold.
		\begin{enumerate}[(i)]
			\item $\mu$ is strongly finite and lower semicomputable.
			\item\label{it:locopt2} For every strongly finite, lower semicomputable outer measure $\nu$ on $\R^n$, $\mu$ dominates $\nu$ on dyadic cubes.
		\end{enumerate}
	\end{definition}

	\begin{theorem}\label{thm:kappalocaloptimal}
		The outer measure $\bkappa$ is locally optimal.
	\end{theorem}

	\begin{proof}
		We rely on machinery created for a different purpose by Case and the first author~\cite{CasLut15}. The LDS coding theorem of~\cite{CasLut15} is a mild generalization of Levin's coding theorem~\cite{jLevi73,jLevi74} that tells us that there is a constant $c\in\N$ such that, for all $r\in\N$ and $Q\in\cQ^{(r)}$,
		\begin{equation}\label{eq:lds}
			\K(Q)\leq \log\frac{1}{\bm(Q)}+\K(r)+c.
		\end{equation}
		To prove the present theorem, it suffices by Observation~\ref{obs:kappasflc} to prove that $\bkappa$ satisfies property~\ref{it:locopt2} of the definition of local optimality. For this, let $\nu$ be a strongly finite, lower semicomputable outer measure on $\R^n$. Define $p_\nu:\Q^n\to[0,\infty]$ by $p_\nu(q)=\nu(\{q\})$ for all $q\in\Q^n$. Then $p_\nu$ is lower semicomputable and $\sum_{q\in\Q^n}p_\nu(q)<\infty$, so the optimality property of $\bm$ tells us that there is a constant $\alpha\in(0,\infty)$ such that, for all $q\in\Q^n$,
		\begin{equation}\label{eq:mopt}
			\bm(q)\geq\alpha p_\nu(q).
		\end{equation}

		Define $\beta:\N\to(0,\infty)$ by
		\[\beta(r)=\alpha 2^{-(\K(r)+c)}\]
		for all $r\in\N$. Then
		\[\lim_{r\to\infty}\frac{\log\frac{1}{\beta(r)}}{r}=\lim_{r\to\infty}\frac{\log\frac{1}{\alpha}+\K(r)+c}{r}=0,\]
		so $\beta(r)=2^{-o(r)}$ as $r\to\infty$. Also, for all $r\in\N$ and $Q\in\cQ^{(r)}$,~\eqref{eq:lds},~\eqref{eq:mopt}, and the countable subadditivity of $\nu$ tell us that
		\begin{align*}
			\bkappa(Q)&=2^{-\K(Q)}\\
			&\geq 2^{-(\K(r)+c)}\bm(Q)\\
			&=2^{-(\K(r)+c)}\sum_{q\in Q\cap\Q^n}\bm(q)\\
			&\geq\beta(r)\sum_{q\in Q\cap\Q^n}p_\nu(q)\\
			&=\beta(r)\sum_{q\in Q\cap\Q^n}\nu(\{q\})\\
			&\geq \beta(r)\nu(Q).
		\end{align*}
		This shows that $\bkappa$ dominates $\nu$ on dyadic cubes, confirming that $\bkappa$ is locally optimal.
	\end{proof}
	
	\begin{corollary}
		A strongly finite, lower semicomputable outer measure on $\R^n$ is locally optimal if and only if it dominates $\bkappa$ on dyadic cubes.
	\end{corollary}

	\begin{lemma}\label{lem:ballcube}
		There is a constant $c\in\N$ such that, for every $r\in\N$, every $r$-dyadic cube $Q\in\cQ^{(r)}$, and every open ball $B\subseteq\R^n$ of radius $2^{-r}$ that intersects $Q$,
		\[|\K(B)-\K(Q)|\leq\K(r)+c.\]
	\end{lemma}
	\begin{proof}
		Lemma 3.5 of~\cite{CasLut15} gives us a constant $c_1\in\N$ such that, for all $r$, $Q$, and $B$ as in the present lemma,
		\[\K(B)\leq\K(Q)+\K(r)+c_1.\]
		Hence it suffices to show that there is a constant $c_2$ such that, for all $r$, $Q$, and $B$ as in the present lemma,
		\[\K(Q)\leq\K(B)+\K(r)+c_2.\]
		Let $M$ be a prefix Turing machine that, on input $\pi_1\pi_2\pi_3$ where $U(\pi_1)=r\in\N$, and $U(\pi_2)=k\in\N$, and $U(\pi_3)=(q_1,\ldots,q_n)\in\Q^n$, outputs the lexicographically $k$\textsuperscript{th} point in the product set
		\[\prod_{i=1}^n\left\{2^{-r}(\lfloor 2^{r}q_i\rfloor-2),2^{-r}(\lfloor 2^{r}q_i\rfloor-1),2^{-r}\lfloor 2^{r}q_i\rfloor,2^{-r}(\lfloor 2^{r}q_i\rfloor+1),2^{-r}(\lfloor 2^{r}q_i\rfloor+2)\right\}.\]

		Let $r$, $Q$, and $B$ be as in the present lemma. Let $q\in B\cap\Q^n$ be such that $\K(q)=\K(B)$. Then there is some point $p=(p_1,\ldots,p_n)\in Q\cap B\cap\Q^n$ such that $|p-q|<2^{1-r}$. Hence $Q$ is the $r$-dyadic cube with address $(\lfloor 2^rp_1\rfloor,\ldots,\lfloor 2^rp_n\rfloor)$, and for each $1\leq i\leq n$,
		\[|\lfloor 2^rq_i\rfloor-\lfloor 2^rp_i\rfloor|\leq 2.\]
		That is, the address of $Q$ belongs to the product set
		\[\prod_{i=1}^n\left\{\lfloor 2^{r}q_i\rfloor-2,\lfloor 2^{r}q_i\rfloor-1,\lfloor 2^{r}q_i\rfloor,\lfloor 2^{r}q_i\rfloor+1,\lfloor 2^{r}q_i\rfloor+2\right\}.\]
		Let $k\leq 5^n$ be the lexicographical index of $Q$'s address within this set, and let $\pi_1$, $\pi_2$, and $\pi_3$ be witnesses to $\K(r)$, $\K(k)$, and $\K(q)$, respectively. Then $M(\pi_1\pi_2\pi_3)\in Q$, so letting $c_M$ be an optimality constant for the machine $M$, we have
		\begin{align*}
			\K(Q)&\leq |\pi_1|+|\pi_2|+|\pi_3|+c_M\\
			&=\K(r)+\K(k)+\K(q)+c_M\\
			&=\K(B)+\K(r)+\K(k)+c_M\,.
		\end{align*}
		Since $k\leq 5^n$, there is some constant $c_3$ such that $\K(k)\leq 2\log(5^n)+c_3$, so the constant
		\[c_2=c_M+2\log(5^n)+c_3\]
		affirms the lemma.
	\end{proof}

	\begin{lemma}
		A strongly finite, lower semicomputable outer measure $\mu$ dominates $\bkappa$ on balls if and only if it dominates $\bkappa$ on dyadic cubes.
	\end{lemma}
	\begin{proof}
		Suppose that $\mu$ dominates $\bkappa$ on balls. Then for every $x\in\R^n$,
		\[\mu(B(x,2^{-r}))=2^{-o(r)}\bkappa(B(x,2^{-r})).\]
		Let $r\in\Z$, and let $Q$ be an $r$-dyadic cube with center $q$, and let $B=B(q,2^{-r-1})$, so that $B\subseteq Q$. Then, applying Lemma~\ref{lem:ballcube},
		\begin{align*}
			\mu(Q)&\geq\mu(B)\\
			&= 2^{-o(r)}\bkappa(B)\\
			&= 2^{-\K(B)-o(r)}\\
			&=2^{-\K(Q)-o(r)}\\
			&=2^{-o(r)}\bkappa(Q),
		\end{align*}
		so $\mu$ dominates $\bkappa$ on dyadic cubes.

		Now suppose that $\mu$ dominates $\bkappa$ on dyadic cubes. Then for every $r$-dyadic cube $Q$,
		\[\mu(Q)=2^{-o(r)}\bkappa(Q).\]
		Let $r\in\Z$, $x\in\R^n$, and $B=B(x,2^{-r})$. Let $Q$ be the $\big(r+\big\lceil\log\sqrt{n}\big\rceil\big)$-dyadic cube containing $x$, so that $Q\subseteq B$.
		Applying Lemma~\ref{lem:ballcube},
		\begin{align*}
			\mu(B)&\geq\mu(Q)\\
			&= 2^{-o(r)}\bkappa(Q)\\
			&= 2^{-\K(Q)-o(r)}\\
			&=2^{-\K(B)-o(r)}\\
			&=2^{-o(r)}\bkappa(B),
		\end{align*}
		so $\mu$ dominates $\bkappa$ on balls.
	\end{proof}

	\begin{corollary}\label{cor:tfaeopt}
		For every strongly finite, lower semicomputable outer measure $\mu$ on $\R^n$, the following three conditions are equivalent.
		\begin{enumerate}[(1)]
			\item $\mu$ is locally optimal.
			\item $\mu$ dominates $\bkappa$ on balls.
			\item For every strongly finite, lower semicomputable outer measure $\nu$ on $\R^n$, $\mu$ dominates $\nu$ on balls.
		\end{enumerate}
	\end{corollary}

	We now have everything we need to prove our main theorem, which is the following generalization of Observation~\ref{obs:kappaloc}.

	\begin{theorem}\label{thm:main}
		If $\mu$ is any locally optimal outer measure on $\R^n$, then for all $x\in \R^n$,
		\begin{equation}\label{eq:dimlocdim}
			\dimloc\mu(x)=\dim(x)
		\end{equation}
		and
		\begin{equation}\label{eq:DimlocDim}
			\Dimloc\mu(x)=\Dim(x).
		\end{equation}
	\end{theorem}

	\begin{proof}
		Let $\mu$ be any locally optimal outer measure on $\R^n$. By Theorem~\ref{thm:kappalocaloptimal} and Corollary~\ref{cor:tfaeopt}, $\bkappa$ dominates $\mu$ on balls, and $\mu$ dominates $\bkappa$ on balls. That is, there exist two function $\beta_1,\beta_2:\N\to(0,\infty)$ such that $\beta_1(r)=2^{-o(r)}$ and $\beta_2(r)=2^{-o(r)}$ as $r\to\infty$, and, for every $r\in\N$ and $x\in\R^n$,
		\[\bkappa(B(x,2^{-r}))\geq \beta_1(r)\mu(B(x,2^{-r}))\]
		and
		\[\mu(B(x,2^{-r}))\geq \beta_2(r)\bkappa(B(x,2^{-r})).\]
		Letting $\beta(r)=\min\{\beta_1(r),\beta_2(r)\}$, we have $\beta(r)=2^{-o(r)}$ as $r\to\infty$ and, for every $r\in\N$ and $x\in\R^n$,
		\[\left|\log\frac{1}{\mu(B(x,2^{-r}))}-\log\frac{1}{\bkappa(B(x,2^{-r}))}\right|\leq\log\frac{1}{\beta(r)}.\]
		It follows that, for all $x\in\R^n$,
		\[\left|\log\frac{1}{\mu(B(x,2^{-r}))}-K_r(x)\right|=o(r)\]
		as $r\to\infty$, whence~\eqref{eq:dimlocdim} and~\eqref{eq:DimlocDim} follow from~\eqref{eq:dim}--\eqref{eq:Dimloc}.
	\end{proof}

    \section{Point-to-Set Principles and Dimensions of Measures}\label{sec:pspdm}

		Local dimensions of measures give rise to global dimensions of measures, which we now briefly comment on. In classical fractal geometry, the global dimensions of Borel measures play a substantial role in studying the interplay between local and global properties of fractal sets and measures. The material in this section is from~\cite{Lutz17}.
	
	\begin{definition}[\cite{Falc97}]
		For any locally finite Borel measure $\mu$ on $\R^n$, the \emph{lower} and \emph{upper Hausdorff and packing dimension} of $\mu$ are
		\begin{align*}
		\ldimH(\mu)&=\sup\big\{\alpha \;\big|\; \mu(\{x \mid \dimloc\mu(x) < \alpha\}) = 0\big\}
		\\\udimH(\mu)&=\inf\big\{\alpha \;\big|\; \mu(\{x \mid \dimloc\mu(x) > \alpha\}) = 0\big\}
		\\\ldimP(\mu)&=\sup\big\{\alpha \;\big|\; \mu(\{x \mid \Dimloc\mu(x) < \alpha\}) = 0\big\}
		\\\udimP(\mu)&=\inf\big\{\alpha \;\big|\; \mu(\{x \mid \Dimloc\mu(x) > \alpha\}) = 0\big\},
		\end{align*}
		respectively.
	\end{definition}
	Extending these definitions to outer measures, we may consider global dimensions of the outer measure $\bkappa$. Since $\bkappa$ is supported on $\Q^n$ and $\dim(x)=0$ for all $x\in\Q^n$,
	\begin{equation}\label{eq:globaldim}
		\ldimH(\bkappa)=\udimH(\bkappa)=\ldimP(\bkappa)=\udimP(\bkappa)=0.
	\end{equation}
	
	The \emph{point-to-set principle}~\cite{LutLut18} expresses classical Hausdorff and packing dimensions in terms of \emph{relativized} algorithmic dimensions. That is, algorithmic dimensions in which the underlying universal Turing machine $U$ is an oracle machine with access to some oracle $A\subseteq\N$. We write $\dim^A(x)$ and $\Dim^A(x)$ to denote the algorithmic dimension and strong algorithmic dimension of a point $x\in \R^n$ relative to $A$.
	\begin{theorem}[\cite{LutLut18}]\label{thm:p2s}
		For every $E\subseteq \R^n$,
		\[\dimH(E)=\adjustlimits\min_{A\subseteq\N}\sup_{{x}\in E}\,\dim^A({x})\]
		and
		\[\dimP(E)=\adjustlimits\min_{A\subseteq\N}\sup_{{x}\in E}\,\Dim^A({x}).\]
	\end{theorem}
	In light of Theorem~\ref{thm:main}, this principle may be considered a member of the family of results, such as Billingsley's lemma~\cite{Bill61} and Frostman's lemma~\cite{Fros35}, that relate the local decay of measures to global properties of measure and dimension. Useful references on such results include~\cite{BisPer17,Hoch14,Matt95}.
	
	Among classical results, this principle is most directly comparable to 
	the \emph{weak duality principle} of Cutler~\cite{Cutl94} (see also~\cite{Falc97}),
	which expresses Hausdorff and packing dimensions in terms of lower and upper pointwise dimensions of measures. For nonempty $E\subseteq\R^n$, let $\Delta(E)$ be the collection of Borel probability measures on $\R^n$ such that $E$ is measurable and has measure 1, and let $\overline{E}$ be the closure of $E$.
	\begin{theorem}[\cite{Cutl94}] For every nonempty $E\subseteq\R^n$,
		\[\dimH(E)=\adjustlimits\inf_{\mu\in\Delta(\overline{E})}\sup_{{x}\in E}\,\dimloc\mu(x)\] 
		and
		\[\dimP(E)=\adjustlimits\inf_{\mu\in\Delta(\overline{E})}\sup_{{x}\in E}\,\Dimloc\mu(x).\]
	\end{theorem}
	By letting $\mathcal{A}=\{\bkappa^A\mid A\subseteq\N\}$ and invoking Observation~\ref{obs:kappaloc}, Theorem~\ref{thm:p2s} can be restated even more similarly as
	\[\dimH(E)=\adjustlimits\inf_{\mu\in\mathcal{A}}\sup_{x\in E}\dimloc\mu(x)\]
	and
	\[\dimP(E)=\adjustlimits\inf_{\mu\in\mathcal{A}}\sup_{x\in E}\Dimloc\mu(x).\]
	Notice, however, that the collections over which the infima are taken in these two results, $\mathcal{A}$ and $\Delta(\overline{E})$, are disjoint and qualitatively very different. In particular, $\mathcal{A}$ does not depend on $E$. Whereas the global dimensions of the measures in $\Delta(\overline{E})$ are closely tied to the dimensions of $E$~\cite{Falc97},~\eqref{eq:globaldim} shows that the outer measures in $\mathcal{A}$ all have trivial global dimensions.

    \section*{Acknowledgments}
	This research was supported in part by National Science Foundation Grants 1445755, 1545028, and 1900716. Some of this work was conducted while the second author was at Rutgers University, the University of Pennsylvania, and Iowa State University. We thank anonymous reviewers for their useful corrections and suggestions.

    \bibliography{oom}
\end{document}